\documentclass[12pt]{article}

\usepackage{amsmath,amssymb,amsfonts,amsthm,setspace,tikz,dsfont, xcolor, array,colortbl,mathrsfs,pifont,wasysym}
\usetikzlibrary{calc}
\usetikzlibrary{positioning}
\usepackage{centernot}
\usepackage{xspace}
\usepackage{mathtools}
\usepackage[scr=boondox,scrscaled=1.05]{mathalfa}
\usepackage[numbers]{natbib}

\usepackage{graphicx}
\usepackage{subcaption}

\usepackage{thmtools,thm-restate}
\usepackage{enumitem}
\usepackage{xfrac} 
\usepackage{scalerel} 

\usepackage[makeroom]{cancel}
\usepackage{bbm}

\usepackage[letterpaper]{geometry}
\definecolor{lam1}{HTML}{900C3F}
\definecolor{lam2}{HTML}{364D73}

\usepackage[colorlinks=true,urlcolor=lam2,linkcolor=lam2,citecolor=lam2]{hyperref}

\onehalfspace

\newtheoremstyle{mytheoremstyle} 
    {\topsep}                    
    {\topsep}                    
    {\itshape}                   
    {}                           
    {\sc}                   
    {.}                          
    {.5em}                       
    {}  

\theoremstyle{mytheoremstyle}

\newtheorem{theorem}{Theorem}
\newtheorem{lemma}{Lemma}

\newtheorem{corollary}[theorem]{Corollary}
\newtheorem{axiom}{Axiom}

\newtheoremstyle{scfont} 
    {\topsep}                    
    {\topsep}                    
    {}                   
    {}                           
    {\scshape}                   
    {}                          
    {.5em}                       
    {\textbf{Axiom \thmnumber{#2}\thmname{#1}}\thmnote{---{\textsc{#3}.}}}

\theoremstyle{scfont}

\theoremstyle{remark}   
\newtheorem{remark}{Remark}
\newtheorem{example}{Example}

\newcounter{ax}

\newcounter{ex}

\usepackage{titlesec}

\titleformat*{\section}{\large\scshape\centering}
\titleformat*{\subsection}{\scshape\centering}
\titleformat*{\subsubsection}{\itshape}
\titleformat*{\paragraph}{\large\bfseries\centering}
\titleformat*{\subparagraph}{\large\bfseries\centering}

\titlespacing*{\section}{0pt}{5.5ex plus 1ex minus .2ex}{3.5ex plus .2ex}

\renewcommand{\emptyset}{\varnothing}
\renewcommand{\phi}{\varphi}
\renewcommand{\implies}{\rightarrow}

\def \A{\mathbf{A}}
\def \B{\mathbf{B}}

\def \<{\langle}
\def \>{\rangle}

\def \s{\succcurlyeq}
\def \sq{\sqsubseteq}
\def \1{\mathbf{1}}
\def \0{\mathbf{0}}

\def\im{\textup{im}}
\def\KO{\mathcal{KO}}
\def\CO{\mathcal{CO}}

\makeatletter
\newcommand{\mylabel}[3]{\def\@currentlabel{#2}\phantomsection\textsc{#3} (\texttt{#2})\label{#1}}
\newcommand{\sclabel}[3]{\def\@currentlabel{\theax}\phantomsection\lowercase{\textsc{#3}} (\textsc{#2})\label{ax:#1}}
\makeatother

\renewcommand{\r}[1]{\hyperref[#1]{\textup{(\texttt{\ref{#1}})}}}
\newcommand{\rax}[1]{\hyperref[ax:#1]{\textup{(\textsc{#1})}}}

\newcommand{\suc}[1]{{\uparrow}(\succ, #1)}
\newcommand{\uc}[1]{{\uparrow}(\s, #1)}
\newcommand{\lc}[1]{{\downarrow}(\s, #1)}
\newcommand{\sig}[1]{\sigma_{#1}}

\DeclareMathOperator*{\sx}{\s^{\hspace{-.2ex}\leftrightarrows}}
\DeclareMathOperator*{\ssr}{\s^{\hspace{-.2ex}\rightarrow}}

\DeclareMathOperator*{\ssl}{\s^{\hspace{-.2ex}\leftarrow}}
\DeclareMathOperator*{\st}{\unrhd}
\DeclareMathOperator*{\stx}{\unrhd^{\hspace{-.2ex}\leftarrow}}
\DeclareMathOperator*{\succt}{\rhd}
\DeclareMathOperator*{\succtx}{\rhd^{\hspace{-.2ex}\leftarrow}}

\DeclareMathOperator*{\cng}{\beta}

\definecolor{darkblue}{rgb}{0,0,0.7}
\definecolor{darkgreen}{rgb}{0,0.4,0}

\usepackage{marginnote}
\newenvironment{change}{\color{darkblue}}{}
\newcommand{\BC}{\begin{change}\marginnote{*}}
\newcommand{\EC}{\end{change}}

\title{Coarse Descriptions and Cautious Preferences}

\author{
Evan Piermont\footnote{Royal Holloway, University of London, Department of Economics, United Kingdom; \href{mailto:evan.piermont@rhul.ac.uk}{\tt evan.piermont@rhul.ac.uk}.} \ and 
Marcus Pivato\footnote{Centre d’\'Economie de la Sorbonne, Universit\'e Paris 1 Panth\'eon-Sorbonne, France; \href{mailto:Marcus.Pivato@univ-paris1.fr}{\tt Marcus.Pivato@univ-paris1.fr}}
} 

\begin{document}

\maketitle

\begin{abstract}
We consider a model where an agent is must choose between alternatives that each provide only an {\em imprecise} description of the world (e.g. linguistic expressions).   The set of alternatives is closed under logical conjunction and disjunction, but not necessarily negation.  (Formally: it is a distributive lattice, but not necessarily a Boolean algebra).   In our main result, each alternative is identified with a subset of an
(endogenously defined) state space, and  two axioms characterize {\em maximin} decision making.  This means:  from the agent's preferences over alternatives,
we derive a preference order on the endogenous state space, such that
 alternatives are ranked in terms of their worst outcomes.  

 \medskip
 
{\bf JEL class:} D81, D86, D73 

\end{abstract}

\section{Introduction}

Any verbal description of an economic situation is necessarily  incomplete.  This creates two problems for decision theory.
First, virtually all  models in economics suppose that agents choose between precisely identified alternatives.  But in practice, agents often choose between {\em verbal descriptions} of alternatives.  For example: when a firm negotiates a legal contract, or a parliament debates legislation, the objects of consideration are written texts.  When a manager in an organization gives  instructions to a subordinate or circulates a memo, her decision  concerns the {\em wording} of the instructions or the memo.  In all these cases, the agent can never provide more than an imprecise verbal description of the  outcome they hope to achieve.   We need a model of decision-making that works with such coarse descriptions.

Second, even when an agent  really does  choose between precisely identified
outcomes, we face a methodological problem.  The analyst's (verbal) description of the agent's decisions will inevitably be imprecise.  This  creates a wedge between an agent's preferences and the analyst's ability to describe, record, or investigate those preferences. To take an elementary example, consider an attempt to inquire into an agent's preference between apples and bananas. A modeler might simply ask her what she would choose between an ``\emph{apple}'' and a ``\emph{banana}'', but these serve only as coarse descriptions that ignore the particularities: their variety, ripeness, whether they have been dropped, their country of origin, etc. So while the agent might prefer the abstract ``\emph{banana}'' to the abstract  ``\emph{apple},'' her preference might be overturned upon further specification: between, say, a ``\emph{185 gram crisp British Gala apple}'' and a ``\emph{slightly over-ripe and mildly dented Chiquita banana}.''

These later descriptions, while surely more cumbersome, form a perfectly valid inquiry. Indeed, we find no theoretical problem with the standard model so long as the objects of choice can eventually be described with sufficient detail so as to decide any preference-relevant query.\footnote{By standard model, we refer to the preponderance of decision theory wherein the choice objects are delineated by some set $X$ and querying the agent's preference between any $x,y \in X$ is costless. Thus, tacitly the elements of $X$ are labeled by some universal set of labels agreed upon by the modeler and the agent.}
But as already noted, this is the exception rather than the rule: many economically relevant outcomes are unboundedly complex, so that any finite description of them can be further resolved to include additional relevant details. 
Put differently: even if the agent's preferences are defined over ideal 
outcomes that fully determine the answer to every potential query
about them, our ability to identify these preferences is limited by what we can describe.  
   
   
   These two problems are part of the impetus for the growing literature on syntactic decision theory \cite[]{Jeffrey65,blume2021constructive,bjorndahl2021language,piermont2023failures,piermont2023vague,guerdjikova2023bounded}, and  motivate the present paper.  
 In this paper, we provide  minimum consistency requirements to ensure that a modeler can perfectly recover an agent's true preference by observing only her preference over coarse descriptions of the outcomes. Importantly,   the modeler need not directly observe the set of objects at all, only the set of descriptions.    Our decision criterion is inspired by the observation that in many of the scenarios described above, decision-makers are cautious and conservative, and assess each coarse description in terms of its {\em worst possible instantiation}.  For example,  when negotiating a legal contract,  an agent regards her counterparty as an adversary, and  considers how each clause in the contract could be weaponized  against her.
In an adversarial legislature, such as the United States Congress or the British House of Parliament, each legislative proposal is criticized by its opponents in terms of its most baleful conceivable consequences.  Before any administrative decision, from launching a major new policy to approving a minor request,  bureaucrats and managers  ask themselves, ``How could this decision go badly wrong? And if it does, will I be blamed?"  

As these examples suggest, in some decisions involving coarse descriptions, an agent will adopt a pessimistic or {\em maximin} decision procedure:   
she will prefer $a$ over $b$
 if the worst possible instantiation of description  $a$  is better than the worst instantiation of $b$.  More generally, in cases where there is no such ``worst instantiation'', she will prefer $a$ over $b$ if, for every instantiation of $a$, there is a {\em worse} instantiation of $b$.  Our main result is an axiomatic characterization of this criterion.
 
  In the study of decision-making under uncertainty, the maximin criterion has been an important alternative to subjective expected utility maximization since the early work of Wald \cite{Wald45}.   But our work differs from Wald's and others in that  the primitives of our model do not contain a state space or any other explicit representation of uncertainty;  there is just a set of descriptions, $\A$, and a preference relation thereon, $\s$. 
Given this minimal structure, we characterize---in two very simple axioms---representations of the form ${\<X, \sigma, \s^*\>}$, where $X$ is a set of outcomes, $\sigma: \A \to 2^X$ determines for each description $a\in \A$, the set of outcomes $\sigma(a) \subseteq X$ that satisfy it, and $\s^*$ is a preference relation on $X$ such that for all $a,b \in \A$:
\begin{equation}
\label{eq:daulOrdINTRO}
a \s b \quad \text{ if and only if } \quad  \forall x_a \in \sigma(a),\  \exists y_b  \in \sigma(b), \text{ such that } x_a \s^* y_b.
\end{equation}
For example, 
if $a =$ ``\emph{apple}'' and $b =$ ``\emph{banana}'', then $\sigma(a) \subseteq X$ is the set of all apples and $\sigma(b)$ all bananas. More detailed descriptions, as above, thus lead to smaller sets of 
outcomes. 
Critically, there will in general be no description of any particular $x \in X$ (that is, there is no $c \in \A$ such that $\sigma(c) = x$).  

The connection between $\s$ and $\s^*$, as governed by \eqref{eq:daulOrdINTRO} is that of a cautious  decision maker. The agent prefers $a =$ ``\emph{apple}'' to $b =$ ``\emph{banana}''  iff for every apple $x_a \in \sigma(a)$ she can find some banana $y_b \in \sigma(b)$ worse than it, that is, such that $x_a \s^* y_b$. Notice this implies $a$ is \emph{strictly} preferred to $b$ exactly when there exists some $y_b \in \sigma(b)$ such that every $x_a \in \sigma(a)$ is strictly preferred  to $y_b$ (by $\s^*$). From this vantage, it is clear that  $\s$ represents maximizing the worst-case guarantee. In particular, should such worst-case outcomes exist (e.g., if $\sigma(a)$ and $\sigma(b)$ were  finite), this criteria would reduce to simply comparing them, i.e., to maximin.

Our key technical result shows that the orders $\s$ and $\s^*$ are {\em dual} to each other, as we now explain.  Just as every description in $\A$ can be identified with the set of outcomes in $X$ that satisfy that description, every outcome in $X$ can be identified with {\em the set of descriptions in $\A$ that it satisfies}.
Formally, for any $x\in X$, let $\tau(x):=\{a\in A$; \ $x\in\sigma(a)\}$.
We  show that $\s$ can be derived from $\s^*$ via  the ``maximin'' formula (\ref{eq:daulOrdINTRO}) iff
  $\s^*$ can be derived from $\s$ via the dual ``maximax'' formula:
\begin{equation}
\label{eq:daulOrdINTRO2}
x \s^* y \quad \text{ if and only if } \quad  \forall b \in \tau(y),\  \exists a  \in \tau(x), \text{ such that } a \s b.
\end{equation}
We require a minimal algebraic structure on $\A$: for any two descriptions $a,b \in \A$, there are further descriptions $a \land b$ and $a \lor b$.
The former, read ``$a$ and $b$,'' corresponds to the requirement that both the descriptions $a$ and $b$ hold, and the latter, read ``$a$ or $b$,''  to the requirement that at least one of the descriptions $a$ or $b$ hold.
 Furthermore, the operations $\land$ and $\lor$ must distribute over one another.\footnote{I.e., $a \land (b \lor c)$ is equivalent to $(a \land b) \lor (a \land c)$ and $a \lor (b \land c)$ is equivalent to $(a \lor b) \land (a \lor c)$.}  The set $\A$ must also contain two ``trivial'' descriptions: one that is satisfied by {\em all} outcomes, and one that is  satisfied by none of them.
 Importantly,  we do {\em not} require elements of $\A$ to have negations or complements.  But  of course we allow this. 

This generality allows us to capture many different environments, where the structure of the descriptions arise from different constraints. The simplest, most easily recognizable case is when $\A$ is  a Boolean algebra, representing those descriptions that can be made in some binary-valued propositional logic. Our results also accommodate more general logical relationships between the descriptions: for example, vague descriptions that can take intermediate truth values can be captured by taking $\A$ to be an MV-algebra, or constructive descriptions that require instantiation can be captured by taking $\A$ to be a Heyting algebra. We can also capture notions of limited awareness, where the entire set of possibilities cannot be described sensibly. For example, it may be reasonable to envision the set of all \emph{apples}, but meaningless to speak of the set of \emph{non-apples}. In such cases, the negation of particular descriptions fails to exist (so that $\A$ is a general distributive lattice).

\vspace{1em}

The rest of the paper is organized as follows.  In Section \ref{S:repr}, 
we state our axioms and the representation theorem described above (Theorem \ref{thm:rep}).  With an additional axiom, we establish a unique {\em minimal} representation (Theorem \ref{thm:uniq}). In the case when $\A$ is a Boolean algebra, we obtain a {\em topological} representation (Corollary \ref{thm:cont_uniq}).
 In Section \ref{S:duality}, we introduce the  formal machinery of dual orders over spectral spaces  and state our general duality result (Theorem \ref{thm:dual}), from  which Theorem \ref{thm:rep} follows.    Finally, all proofs are in Section \ref{S:proofs}.

\section{\label{S:repr}Representation Theorem and Discussion}

\subsection{Notation}

A {\em bounded distributive lattice} is a tuple $(\A, \lor, \land, \0, \1)$, where $\A$ is a set, $\lor$ and $\land$ are binary operators and $\0$ and $\1$ are elements such that $\land$ and $\lor$ are commutative, associative, distribute over each other, and satisfy the absorption laws: $a = a \lor (a \land b)$ and $a = a \land (a \lor b)$. When it is not confusing to do so, we abuse notation and refer to the lattice by its carrier set, $\A$.

There is a canonical ordering, $\sq$, over $\A$ defined as follows: $a \sq b$ if and only if $a = a \land b$. That $\A$ is bounded requires $\0 \sq a \sq \1$ for all $a \in \A$. Note for any set $X$, $(2^X, \cup,\cap,\emptyset,X)$ forms a bounded distributive lattice.
If $(\A, \lor_{\A}, \land_{\A}, \0_{\A}, \1_{\A})$ and $(\mathbf{B}, \lor_{\mathbf{B}}, \land_{\mathbf{B}}, \0_{\mathbf{B}}, \1_{\mathbf{B}})$ are lattices, then $h: \A \to \mathbf{B}$ is a \textit{lattice homomorphism} if $h(\0_\A) = \0_{\mathbf{B}}$, $h(\1_\A) = \1_{\mathbf{B}}$ and for all $a, b \in \A$, $h(a \land_\A b) = h(a) \land_{\mathbf{B}} h(b)$, $h(a \lor_\A b) = h(a) \lor_{\mathbf{B}} h(b)$. An invertible homomorphism is called an {\em isomorphism}.

\subsection{Model}

Let $(\A, \lor, \land, \0, \1)$ be bounded distributed lattice capturing the set of \emph{descriptions}. That is, each $a\in \A$ is a description of choice outcomes that may be true or false of any given outcome. The ordering $\sq$ captures {\em specificity}: $a \sq b$ exactly when $a$ is a more specific description that $b$ (i.e., $a$ describes a further specification of the possible outcomes that $b$ describes). The primitive of the model is a complete and transitive relation $\s$ over $\A$.

Consider the following simple axioms:

\begin{axiom} 
\label{ax:upper}
For all $a,b \in \A$, if $a \sq b$ then $a \s b$.
\end{axiom}

Axiom \ref{ax:upper} states that preferences are monotone in the specificity ordering. Recall, we are considering a cautious decision maker who evaluates descriptions based on their worst-case guarantee. Since more specific descriptions reduces the set of possibilities, they will have higher lower bounds. 




\begin{axiom}
\label{ax:union}
For all $a,a',b \in \A$, if $a \succ b$ and $a' \succ b$ then $(a \lor a') \succ b$
\end{axiom}

Axiom \ref{ax:union} states that strict preference `distributes' over $\lor$. Again the motivation for this comes from the cautious nature of the decision maker.  Recall that $a \succ b$ when the agent knows that all outcomes consistent with $a$ dominate some outcome consistent with $b$.  Further, the outcomes consistent with $(a \lor a')$ are exactly the union of those outcomes consistent with each description individually. So if all $a$-outcomes dominate some $b$-outcome, and all $a'$-outcomes dominate some other $b$-outcome, then all $(a \lor a')$-outcomes dominate the least preferred of these two $b$-outcomes.

These two axioms yield considerable structure:

\begin{theorem}
\label{thm:rep}
The following are equivalent:
\begin{enumerate}[label=\textup{(\roman*)}]
\item\label{t.ax1} $\s$ satisfies Axioms \ref{ax:upper} and \ref{ax:union} 
\item\label{t.rep} There exist some $\<X, \sigma, \s^*\>$, where $X$ is a set, $\sigma: \A \to 2^X$ is a lattice homomorphism,
 and $\s^*$ is a weak ordering of $X$ such that for all $a,b \in \A$:
\begin{equation}
\label{eq:thm_rep}
\tag{$\dag$}
a \s b \quad \text{ if and only if } \quad  \forall x_a \in \sigma(a),\  \exists y_b  \in \sigma(b), \text{ such that } x_a \s^* y_b.
\end{equation}
\end{enumerate}
\end{theorem}

The implication \ref{t.rep} $\implies$ \ref{t.ax} is proven directly in Lemma \ref{lem:converse}; the other direction is the consequence of our duality result in the next section.


An additional axiom implies the existence of a unique minimal set $X$ of alternatives from which such a representation can arise. The axiom states that if $a$ and $a'$ are trivialized by the same set of further specifications, then they are indifferent (that is, for each $b$,  $a \land b$ is treated by the agent as if it is impossible iff $a' \land b$ is as well). 


\begin{axiom}
\label{ax:cong}
For $a,a' \in \A$ such that  $a \land b \sim \0 $ iff $a' \land b \sim \0$ for all $b \in \A$, we have $a' \sim a$.
\end{axiom}


This axiom requires that differences in preferences can be accounted for by the set descriptions deemed impossible, those that are indifferent to $\0$. While Axiom \ref{ax:cong} has a technical character, it should be pointed out that it is not particularly strong, and indeed, is implied by the other axioms whenever $\A$ has complements.

\begin{remark}
If $\A$ has relative complements and $\s$ satisfies Axioms \ref{ax:upper} and \ref{ax:union}, then $\s$ satisfies Axiom \ref{ax:cong}.
\end{remark}

\begin{proof}
Assume that $\{ b \in \A \mid a \land b \sim \0\} = \{ b \in \A \mid a' \land b \sim \0\}$. 
 Denote by $c$ the relative complement of $a$ with respect to $a'$: that is, $a \lor c = a \lor a'$ and $a \land c = \0$.  Then $a' \land c=c $.\footnote{{\em Proof.}  $a'\land c = \0 \lor(a'\land c) =_* (a\land c)\lor (a'\land c)
=(a\lor a')\land c=_* (a\lor c)\land c=(a\land c)\lor (c\land c)=_*\0\land c=c$,
where each $*$ equality is by the definition of relative complement.}

Since $a \land c \sim \0$, the assumption implies $a' \land c \sim \0$;
 in other words $c\sim \0$. 
 By Axiom \ref{ax:upper}, $a' \s a \lor a' = a \lor c$. Thus the contra-positive of Axiom \ref{ax:union} implies $a' \s a$ or $a' \s c$. Since $c \sim \0 \s a$, by Axiom \ref{ax:upper}, we can conclude that $a' \s a$. A symmetric argument yields $a \s a'$.
\end{proof}

With this axiom, we can establish a minimal representation:

\begin{theorem}
\label{thm:uniq}
If $\s$ satisfies Axioms \ref{ax:upper}, \ref{ax:union} and \ref{ax:cong}, then there exists a unique (up to isomorphism) minimal representation, $\<X, \sigma, \s^*\>$. Specifically, if $\<\hat X, \hat \sigma, \hat \s\>$ is any other representation then there exists a (surjective) homomorphism $h: \im(\hat\sigma) \to \im(\sigma)$ such that $\sigma = h \circ \hat\sigma$.
 \end{theorem}

To see how the representation $\<X, \sigma, \s\>$ is minimal:  assume that for some other representation $\<\hat X, \hat \sigma, \hat \s\>$, two descriptions $a,b \in \A$ turn out to be equivalent, in the sense that $\hat\sigma(a) = \hat\sigma(b)$. This is a simplification as it allows the modeler to exclude superfluous objects that might satisfy $a$ but not $b$ or vice-versa. In such a case, since $\sigma = h \circ \hat\sigma$, it must be that $\sigma(a) = \sigma(b)$ as well. So any simplification that can be made in any representation is made in the minimal one.

For the special case of complete Boolean algebras, this minimality condition is sharpened considerably.\footnote{Recall that a lattice $\A$ is a {\em Boolean algebra} if it has a ``negation'' operator $\neg$ such that (i) 
$\neg (a\land b)=(\neg a)\lor(\neg b)$ and 
$\neg (a\lor b)=(\neg a)\land(\neg b)$ for all $a,b\in\A$ and (ii) 
$\neg\neg a=a$ for all $a\in\A$. \newline\mbox{}\quad  A Boolean algebra $\A$ is {\em complete} if any subset of $\B\subseteq\A$ has a 
{\em supremum} (a unique minimal  $a\in\A$ such that $b \sq a$ for all $b \in\B$) and an {\em infimum} (a unique maximal  $a\in\A$ such that $a \sq b$ for all $b \in\B$).}
 Call $\<X,\tau, \sigma, \s\>$ a \emph{topological representation} if $(X,\tau)$ is a Stone space\footnote{That is: $X$ is compact, Hausdorff and totally disconnected.} and $\sigma$ is clopen-valued (i.e., $\sigma(a)$ is clopen for all $a \in \A$). By focusing on topological representations of complete Boolean algebras we can recast minimality in terms of the objects themselves:

\begin{corollary}
\label{thm:cont_uniq}
Let $\A$ be a complete Boolean algebra. If $\s$ satisfies Axioms \ref{ax:upper} and \ref{ax:union}, and $\{a \in \A \mid a \sim \0\}$ has a maximal element, then there exists a unique (up to isomorphism) minimal topological representation, $\<X, \tau, \sigma, \s^*\>$. Specifically, if $\<\hat X, \hat\tau, \hat \sigma, \hat \s\>$ is any other topological representation then there exists a continuous injection $i: X \to \hat X$ such that $\sigma = i^{-1} \circ \hat\sigma$.
 \end{corollary}
 
 \noindent {\sc Remark.} As noted earlier, one interpretation of the ``maximin'' formula (\ref{eq:thm_rep})  is of an agent who chooses a description $a$ in $\A$ expecting a hostile ``adversary'' to then choose an outcome in $X$ satisfying $a$.  An alternative scenario, where the agent expects that she herself (or an ally) will choose an outcome satisfying $a$, would yield a ``maximax'' representation:
$a$ is better than $b$ if, for every outcome $x_b$ satisfying $b$, there is some outcome $x_a$ satisfying $a$ such that $x_a \s x_b$. 
It is easy to prove  versions of Theorems \ref{thm:rep} and \ref{thm:uniq} and Corollary
\ref{thm:cont_uniq} that
axiomatically characterize this alternative representation by simply ``reversing the polarity'' of our three axioms.
The interpretation here  parallels the \emph{menu choice} literature {\em \`a la} Kreps \cite{kreps1979representation}; the reversed version of Axiom \ref{ax:upper} take the place of Kreps' \emph{preference for flexibility} axiom.

\section{\label{S:duality}A General Duality Result}

In this section, we present the main technical result of the paper. We show that the two Axioms presented above ensure a tight duality between preference orders over a distributive lattice $\A$ and orders over
a canonical set of objects called its \emph{spectral space}.  
This duality immediately proves Theorem \ref{thm:rep} by providing a canonical representation. To state our result, we must first introduce the spectrum of a lattice.

A subset $F \subseteq \A$ is called a \emph{filter} if 
is non-empty, $\sq$-upwards-closed, and for any $a,b\in \A$ if $a \in F$ and $b \in F$, then $a \land b \in F$.
A filter is called \emph{proper} if $F \neq \A$ (equivalently, $\0 \notin F$). A filter $F$ is called \emph{prime}
if $F$ is proper, and, $a \lor b \in F$ implies $a \in F$ or $b \in F$.

The {\em spectrum} of $\A$ is the set
  $X_\A = \{F \subset \A \mid F \text{ is a prime filter}\}$.   For all $a\in \A$,
  let $\sig{\A}(a) $ be the set of  prime filters  containing $a$; formally,
  $\sig{\A}(a) := \{F \in X_\A \mid a \in F\}$.  This defines a function
  $\sig{\A}: \A \to 2^{X_\A}$.  The image of $\sigma_\A$ acts as a base for a topology on $X_\A$ called the \emph{spectral topology}.  Endowed with this topology, $X_\A$ is called the {\em spectral space} of $\A$.  
  Let $\KO(X_\A)$ be the set of all compact-open subsets of $X_\A$.
  Then $\KO(X_\A)$ is a lattice under union and intersection, and the function
  $\sigma_\A$ is a lattice isomorphism between $\A$ and $\KO(X_\A)$
  \cite[Corollary II.3.4, p.66]{Johnstone82}.\footnote{Note that in Johnstone's \cite{Johnstone82} definition, ``lattices'' are always bounded; see \S I.1.4, pp.2-3.} 
When $\A$ is a Boolean algebra, the spectral topology will be Hausdorff, and compact-open sets coincide with clopen sets; this is not in general true, as the following example shows:

\begin{example}
Take $\A = \{0,\frac12, 1\}$ endowed with the natural ordering. There are two prime filters: $x = \{1\}$ and $y =\{\frac12,1\}$, so the spectrum of this lattice is $X_\A = \{x,y\}$ with topology $\{\emptyset, \{x\}, \{x,y\}\}$. Notice this topology is not Hausdorff; consequently, compact sets need not be closed. In particular, the set $\{x\}$ is compact (since any open cover trivially has a finite sub-cover) and open, but it is not closed, as it's complement is not contained in the topology.
\end{example}


For technical reasons, we here consider a preference relation over $\A \setminus \{\0\}$, that is, excluding the bottom element. 
As evidenced by the proof of Theorem \ref{thm:rep}, this assumption is without loss of generality in applications: it is always possible to identify the set of elements indifferent to $\0$ and apply our duality result to the resulting quotient lattice to obtain a faithful representation.

Given a distributive lattice $(\A, \lor, \land, \0, \1)$  and a weak order $\s$ over $\A \setminus \{\0\}$, we can define a relation $\ssr$ over $X_\A$ as follows
\begin{equation}
\label{eq:daulOrd1}
\tag{$\rightarrow$}
F \ssr G \text{ if and only if } \forall b \in G, \exists a \in F, a \s b
\end{equation}
and likewise, from $\s$ over $X_\A$, we can define a relation $\ssl$ over $\A \setminus \{\0\}$ as follows
\begin{equation}
\label{eq:daulOrd2}
\tag{$\leftarrow$}
a \ssl b \text{ if and only if } \forall F \in \sig{\A}(a), \exists G  \in \sig{\A}(b), F \s G
\end{equation}

We are interested in ensuring that $\s$ coincides with $\sx$, that is, that no information is lost from the process of moving back and forth between the lattice $\A$ and its spectral representation $X_\A$.
To better understand this requirement, notice that  $a \sx b$, is equivalent to, by plugging  \eqref{eq:daulOrd1}  into \eqref{eq:daulOrd2},
 \begin{equation}
 \label{eq:needed1}
 \tag{$\leftrightarrows$}
\forall F \in \sig{\A}(a), \exists G \in \sig{\A}(b) \ \  \text{ such that } \ \ \forall b' \in G, \exists a' \in F  \ \ \text{ such that } \ \ a' \s b'.
\end{equation}



\begin{theorem}
\label{thm:dual}
The following are equivalent:
\begin{enumerate}[label=\textup{(\roman*)}]
\item\label{t.ax} $\s$ satisfies Axioms \ref{ax:upper} and \ref{ax:union}.
\item\label{t.filt} 
$a \s b$ if and only if there exists some $G \in \sig{\A}(b)$ such that $a \s b'$ for all $b' \in G$.
\item\label{t.dual} $a \s b$ if and only if  $a \sx b$.
\end{enumerate}
\end{theorem}

\section{\label{S:proofs}Proofs}

\subsection{Preliminaries}

Ideals are order-dual to filters: Specifically is subset $I \subseteq \A$ is called an \emph{ideal} if 
is non-empty, a down-set, and for any $a,b\in \A$ if $a \in I$ and $b \in I$, then $a \lor b \in I$.
A ideal is called \emph{proper} if $I \neq \A$, and called \emph{prime}
if it is proper, and, $a \land b \in I$ implies $a \in I$ or $b \in I$.

A \textit{lattice congruence} on $\A$ is an equivalence relation $\cng$ on $\A$ such that for all $a, b, a',b' \in \A$:
If $a \cng b$ and $a' \cng b'$, then $(a \land a') \cng (b \land b')$ and $(a \lor a') \cng (b \lor b')$. Let $[a]_{\cng}$ denote the equivalence class containing $a$. The set of $\cng$-equivalence classes forms a distributive lattice under the inherited operations, denoted $\sfrac{\A}{\cng}$. 

For a homomorphism $h: \A \to \hat\A$, let $\ker h = \{(a,b) \in \A \times \A \mid h(a) = h(b)\}$ denote its kernel.  (The binary relation defined by $\ker h$  is clearly an equivalence relation on $\A$.) 
The map $h_{\cng}: a \mapsto [a]_{\cng}$ is a homomorphism with kernel equal to $\cng$ \cite[Lemma 6.8, p.133]{davey2002introduction}.
Moreover, for any homomorphism $h$, $\ker h$ is a lattice congruence and $\sfrac{\A}{\ker h}$ is isomorphic to $\im(h)$ \cite[Theorem 6.9, pp.133-134]{davey2002introduction}.
\begin{lemma}
\label{lem:converse}
Let $X$ be a set and $\st$ a complete and transitive relation on $X$. Then for any distributive lattice $\A$ and lattice-homomorphism $\sigma: \A \to 2^X$ the relation $\stx$ defined by 
\begin{equation}
\label{eq:thm_rep_proof}
\tag{$\dag$}
a \stx b \quad \text{ if and only if } \quad  \forall x \in \sigma(a),\  \exists y  \in \sigma(b), \text{ such that } x \st y,
\end{equation}
satisfies Axioms \ref{ax:upper} and \ref{ax:union}.  
\end{lemma}

\begin{proof}
Let $a,b \in \A$, with $a \sq b$.  If $x \in \sig{\A}(a)$, then also $x \in \sig{\A}(b)$ because  $\sig{\A}(b)\subseteq\sig{\A}(b)$;  by reflexivity $x \st x$ and so from \eqref{eq:thm_rep_proof} we have $a \stx b$, as needed for Axiom \ref{ax:upper}.

Denote by $\succt$ and $\succtx$ the strict components of $\st$ and $\stx$, respectively. 
For some $a,a',b \in \A$, assume $a \succtx b$ and $a' \succtx b$; from \eqref{eq:thm_rep_proof} this implies
\begin{align}
\label{eq:sqneed1}&\exists y \in \sig{\A}(b), \forall x \in \sig{\A}(a): \quad x \succt y, && \text{ and,} \\
\label{eq:sqneed2}&\exists y' \in \sig{\A}(b), \forall x' \in \sig{\A}(a'): \quad x' \succt y'.
\end{align}
Denote by $y''$ the $\st$-worse element between $y$ and $y'$. Now, let $x'' \in \sig{\A}(a \lor a') = \sig{\A}(a) \cup \sig{\A}(a')$. Thus, either $x'' \in \sig{\A}(a)$, in which case $x'' \succt y \st y''$ by \eqref{eq:sqneed1}, or $x'' \in \sig{\A}(a')$,  in which case $x'' \succt y' \st y''$ by \eqref{eq:sqneed2}. Since $x''$ was arbitrary, we have established that $\exists y'' \in \sig{\A}(b), \forall x'' \in \sig{\A}(a \lor a'), x'' \succt y''$, or via \eqref{eq:thm_rep_proof}, that $(a \lor a') \succtx b$, as needed for Axiom \ref{ax:union}.
\end{proof}

\subsection{Proof of Theorem \ref{thm:dual}}

\begin{proof}
\ref{t.ax} $\Rightarrow$ \ref{t.filt}
For all $a\in\A$, let $\suc{a} = \{c \in \A \mid c \succ a\} \cup \{\0\}$ denote the strict upper contour set of $a$ (union the bottom element) and $\lc{a} = \{c \in \A \mid a \s c\}$ the weak lower contour set of $a$. It follows from completeness that $\lc{a} = \A \setminus (\suc{b})$.

Now notice that $\suc{a}$ is a proper ideal of $\A$.
Indeed, if $c \in \suc{a}$ and $c' \sq c$, then $c' \s c \succ a$ by Axiom \ref{ax:upper}, and therefore $c' \in \suc{a}$ by transitivity. Likewise, if $c \in \suc{a}$ and $c' \in \suc{a}$ then by Axiom \ref{ax:union}, $(c \lor c')  \in \suc{a}$. Finally, since $a \sq \1$, we have $a \s \1$ by Axiom \ref{ax:upper}, and therefore $\1 \notin \suc{a}$. Finally, $\0 \in \suc{a}$, so $\suc{a}$ is non-empty.

Now it is a straightforward consequence of the Boolean prime ideal theorem for distributive lattices that every ideal is the intersection of all prime ideals containing it (see for example Exercise 3.1.17 of \cite{gehrke2024topological}).

Symbolically, let 
$$\mathcal{I}_a = \{ I \subset \A \mid I \text{ prime ideal of } \A, \suc{a} \subseteq I\}.$$
Then we have
$$\suc{a} = \bigcap_{I \in \mathcal{I}_a} I,$$
implying that
$$\lc{a} = (\bigcap_{I \in \mathcal{I}_a} I)^c =  \bigcup_{I \in \mathcal{I}_a} I^c.$$
Now assume \ref{t.ax}.   The ``$\Leftarrow$'' direction of \ref{t.filt}
is trivial (because if $G \in \sig{\A}(b)$  then $b \in G$).  To prove the ``$\Rightarrow$'' direction of \ref{t.filt}, suppose $a \s b$. Then  $b \in \lc{a}$ so there exists some $I \in \mathcal{I}_a$ such that $b \in I^c$. Now since $I$ is a prime ideal, $I^c$ is a prime filter, so $I^c\in\sig{\A}(b)$. 
Moreover, $I^c \subseteq \lc{a}$, establishing the claim.

\bigskip

\ref{t.filt} $\Rightarrow$ \ref{t.dual} 
First, let $a \s b$. We must establish \eqref{eq:needed1}.
Now since $a \s b$, by \ref{t.filt} there exists some $G^\dag \in \sig{\A}(b)$ such that $a \s b'$ for all $b' \in G^\dag$. So let $F \in \sig{\A}(a)$. We have that for $G^ \dag$, and all $b' \in G^\dag$, $a \s b'$. Since $F$ was arbitrary, this establishes \eqref{eq:needed1}.

Now, let $a \centernot{\s} b$, i.e., let $b \succ a$. We must show that $a  \ {\centernot{\s}^{\hspace{-.2ex}\rightleftarrows}} \ b$.  In other words, 
 \begin{equation}
 \label{eq:needed2}
\exists F \in \sig{\A}(a)  \ \  \text{such that}  \ \  \forall G \in \sig{\A}(b), \   \exists b' \in G  \ \  \text{such that}   \ \ \forall a' \in F,   \ \  b' \succ a'.
\end{equation}

Since $a \s a$, we have by \ref{t.filt} that there exists some $F^\dag \in \sig{\A}(a)$ such that $a \s a'$ for all $a' \in F^\dag$. Moreover, since $a \centernot{\s} b$ we have by \ref{t.filt} that for all $G \in \sig{\A}(b)$ there exists some $b' \in G$ such that $b' \succ a$.  But $a \s a'$ for all  $a' \in F^\dag$.  So by transitivity, 
 for   all $G \in \sig{\A}(b)$,  there exists some $b' \in G$ such that $b' \succ  a'$ for all  $a' \in F^\dag$, establishing \eqref{eq:needed2}.

\ref{t.dual} $\Rightarrow$ \ref{t.ax}
This is an immediate consequence of Lemma \ref{lem:converse}.
%
\end{proof}


\subsection{Proof of Theorems \ref{thm:rep} and \ref{thm:uniq}}

\begin{proof}
Let $\s$ over $\A$ satisfy Axioms \ref{ax:upper} and \ref{ax:union}.
First, note that for all $b \in \A$, $\uc{b}$ is an ideal of $\A$. That $\uc{b}$ is downward closed is an immediate consequence of Axiom \ref{ax:upper}. To see that it is $\lor$-closed, assume to the contrary that $a, a' \in \uc{b}$ but $a \lor a' \notin \uc{b}$. Then by transitivity and completeness, we have $a \s b \succ a \lor a'$ and $a' \s b \succ a \lor a'$. Applying Axiom \ref{ax:union}, we obtain $a \lor a' \succ a \lor a'$, a contradiction.

Let $I = \{a \in \A \mid a \sim \0\} = \uc{\0}$. 
Then $I$ is an ideal by the previous paragraph. 
 Suppose there is a congruence $\cng$ on $\A$ such that
\begin{enumerate}[label=\textup{(\roman*)},itemsep=-1ex,topsep=-.5ex]
\item\label{pf.cng1} $a \cng b$ implies $a \sim b$, and
\item\label{pf.cng2} $[\0]_{\cng} = I$.
\end{enumerate}
\medskip
(We will later show that such a congruence exists.) 
Let $\A_{\cng}$ denote the quotient lattice $\sfrac{\A}{\cng}$, and we can define $\s_{\cng}$ over $\A_{\cng} \setminus [\0]_{\cng}$ in the obvious way: $[a]_{\cng} \s_{\cng} [b]_{\cng}$ if and only if $a \s b$. This is a well-defined weak order by property \ref{pf.cng1}.

The relation $\s_{\cng}$ satisfies Axioms \ref{ax:upper} and \ref{ax:union}; Axiom \ref{ax:union} is obvious. To see Axiom \ref{ax:upper}, let $[a]_\beta \sqsubseteq_{\beta} [b]_\beta$, thus by definition, $[a \land b]_\beta = [a]_\beta \land_\beta [b]_\beta = [a]_\beta$. From \ref{pf.cng1} that $a \land b \sim a$. So by Axiom \ref{ax:upper} (for $\s$) it follows that $a \s b$, and so too that $[a]_\beta \s_{\beta} [b]_\beta$.

As such, by Theorem \ref{thm:dual}, $\s_{\cng}^{\hspace{-.2ex}\leftrightarrows}$ coincides with $\s_{\cng}$. 
It is straightforward to show that $\<X_{\A_{\cng}}, \sig{\A_{\cng}}~\circ~h_\beta, \s_{\cng}^{\hspace{-.2ex}\rightarrow}\>$ is a representation of $\s$ as in \eqref{eq:thm_rep}.\footnote{Notice that for $a \in I$, $\sigma(a) \mapsto \emptyset$, so the required preference holds vacuously.}

To complete the proof to Theorem \ref{thm:rep}, it suffices to produce such a congruence relation. Take the following: $a \cng' b$ if and only if there exists some $c,c' \in I$ such that $b \sq a \lor c$ and $a \sq b \lor c'$.

To see that $\cng'$ is a congruence, let $a_1,a_2,b_1,b_2\in\A$, and suppose that
$a_1\cng' b_1$ and $a_2\cng' b_2$.  Thus, there exist $c_1,c'_1,c_2,c'_2\in I$ such that $b_1 \sq a_1 \lor c_1$, \  $a_1 \sq b_1 \lor c'_1$, \ 
$b_2 \sq a_2 \lor c_2$, and $a_2 \sq b_2 \lor c'_2$.
We must show that
$(a_1\lor a_2) \cng' (b_1\lor b_2)$ and
$(a_1\land a_2) \cng' (b_1\land b_2)$.  
To see the former, observe that
$a_1\lor a_2 \sq (b_1 \lor c'_1)\lor(b_2 \lor c'_2)= (b_1\lor b_2) \lor C'$,
where $C'=c'_1\lor c'_2$.
Likewise, $b_1\lor b_2 \sq (a_1\lor a_2) \lor C$, where $C=c_1\lor c_2$.
But $C, C'\in I$, because  $c_1,c_2,c'_1,c'_2\in I$  and $I$ is $\lor$-closed (being an ideal). Thus,
$(a_1\lor a_2) \cng' (b_1\lor b_2)$. 

Meanwhile, 
$a_1\land a_2 \sq (b_1 \lor c'_1)\land(b_2 \lor c'_2)= (b_1\land b_2) \lor C'$,
where $C' = (b_1\land c'_2)  \lor (c'_1\land b_2) \lor (c'_1\land c'_2)$.
Each of the three disjuncts is in $I$ because $c_1',c'_2\in I$  and $I$ is a downset.  Thus, $C'\in I$ because $I$ is $\lor$-closed.
Likewise, $b_1\land b_2 \sq (a_1\land a_2) \lor C$,
where $C = (a_1\land c_2)  \lor (c_1\land a_2) \lor (c_1\land c_2)\in I$.
Thus, $(a_1\land a_2) \cng' (b_1\land b_2)$.  

To see that $\cng'$  satisfies \ref{pf.cng2}, let $a\in\A$ be arbitrary,
and note that $\0\sq a\lor c$ for
any $c\in I$.  Meanwhile, for any $c'\in I$, we have $\0\lor c'=c'$, 
so $(a\sq   \0\lor c') \Leftrightarrow (a\sq c')$. 
Thus, $(a\in [\0]_{\cng'}) \Leftrightarrow (a\cng' \0)
\Leftrightarrow (a\sq c'  \ \mbox{for some}  \ c'\in I)\Leftrightarrow (a\in I)$, where
the last step is because $I$ is a downset.

 To see that $\cng'$ satisfies \ref{pf.cng1}, let $a \succ b$. Then for all $c \in I$, $c \sim \0 \s a \succ b$. By Axiom \ref{ax:union}, $a \lor c \succ b$, and thus 
by the contrapositive of Axiom \ref{ax:upper}, 
it cannot be that $b \sq a \lor c$; thus, $a \cng' b$ cannot hold.

Finally, to prove Theorem \ref{thm:uniq}, define the coarser congruence on $\A$ as follows: $a \cng'' b$ if and only if $\{c \in \A \mid a \land c \in I\} = \{c \in \A \mid b \land c \in I\}$. 
It is again straightforward to verify that $\cng''$ is a lattice congruence satisfying \ref{pf.cng2} (see Proposition 2.2 of \cite{barzegar2019quotient}). It satisfies \ref{pf.cng1} on account of  Axiom \ref{ax:cong}. So $\<X_{\A_{\cng''}}, \sig{\A_{\cng''}}\circ h_{\cng''}, \s_{\cng''}^{\hspace{-.2ex}\rightarrow}\>$ is a representation. We claim it is minimal. 

Let $\<\hat X, \hat \sigma, \hat \s\>$ be any other representation. By the well known relationships of congruence and homomorphisms (i.e., the `first and second isomorphism theorems'; for specific accounts see Theorems  1.5 and 1.6 of \cite{gratzer2023congruences}) it suffices to show that $\ker \hat\sigma \subseteq \ker \sig{\A_{\cng''}}\circ h_{\cng''}$.


So assume that $(a,b) \in \ker \hat\sigma$, or, that $\hat\sigma(a) = \hat\sigma(b)$. Thus, for any $c \in \A$, we have $\hat\sigma(a\land c) = \hat\sigma(a) \cap \hat\sigma(c)= \hat\sigma(b) \cap \hat\sigma(c) = \hat\sigma(b \land c)$: by \eqref{eq:thm_rep}, $a \land c \sim \0$ if and only if $b \land c \sim \0$. Thus $a \cng'' b$, or, $(a,b) \in \ker h_{\cng''} \subseteq \ker \sig{\A_{\cng''}}\circ h_{\cng''}$.
 \end{proof}

\subsection{Proof of Corollary \ref{thm:cont_uniq}}

\begin{proof}
 Let $\<X, \sigma, \s^{\hspace{-.2ex}\rightarrow}\>$ 
be the minimal representation constructed in the Proof of Theorem \ref{thm:uniq}
on the basis of the congruence $\cng''$ (we omit the $\cng''$-subscript for cleanliness).
Endow $X$ with the spectral topology to form a topological representation.
Let $\CO(X)$ denote the set of all clopen subsets of $X$.  This is a Boolean algebra under union, intersection, and complementation, and 
  $\CO(X)=\KO(X)$, because $X$ is compact and Hausdorff.  Thus, $\CO(X) = \im(\sigma)  \cong \A_{\cng''}$ \cite[Theorem 31, Ch. 34]{GivantHalmos08}.

Since $\A$ is a complete Boolean algebra and $[\0]_{\cng''}$ is a principal ideal (it is principal by the assumption in the theorem, and was shown to be an ideal by the proof of Theorem \ref{thm:rep}), $\A_{\cng''}$ (hence, $\CO(X)$) is also a complete Boolean algebra.\footnote{{\em Proof.}
Any principal ideal is obviously a {\em complete} ideal ---i.e. it is closed under
suprema.  The quotient map defined by any complete ideal is a {\em complete} homomorphism ---i.e. one which preserves all infima and suprema
\cite[Theorem 20, Ch.24] {GivantHalmos08}.  Thus, the quotient
of a complete Boolean algebra by such an ideal is complete.}

Now let $\<\hat X, \hat\tau, \hat \sigma, \hat \s\>$ be any other topological representation, and 
let $\CO(\hat X)$ denote the algebra of clopen subsets of $\hat X$.  Theorem \ref{thm:uniq} ensures the existence a surjective homomorphism
$h:  \im(\hat \sigma) \to \CO(X)$
such that $\sigma = h \circ \hat\sigma$. By the Sikorski extension theorem,\footnote{ See Theorem 33.1 of \cite{sikorski1969boolean} or Theorem 5 of \cite{GivantHalmos08}.} $h$ can be extended to a (surjective) homomorphism $h^*: \CO(\hat X) \to \CO(X)$. By Stone duality, there exists a continuous and injective $i: X \to \hat X$ such that $h^*(B) = i^{-1}(B)$ for all $B \in \CO(\hat X)$.
 \end{proof}

\singlespace
\small
\bibliographystyle{plain}

\bibliography{do.bib}

\end{document}